\documentclass[11pt]{article}
\usepackage{amssymb}
\usepackage{amsmath,amsthm}
\usepackage{color}

\usepackage{fullpage}
\usepackage[dvips]{graphicx}

\newtheorem{properties}{Properties}
\newcommand{\commentout}[1]{}

\newcounter{algline}

\newtheorem{definition}{Definition}[section]
\newtheorem{theorem}{Theorem}[section]

\newtheorem{lemma}[theorem]{Lemma}

\newtheorem{observation}[theorem]{Observation}

\begin{document}

%MN2 changed the title
\title{The Family Holiday Gathering Problem\\ or\\ Fair and Periodic Scheduling of Independent Sets % Towards Alleviating the Stress
}
%\author{
%Amihood Amir\inst{1}\and
%Oren Kapah\inst{2}\and
%Tsvi Kopelowitz\inst{3} \and
%Moni Naor\inst{4}\and
%Ely Porat\inst{2}}
%\institute{Bar-Ilan University and Johns Hopkins University, \email{amir@cs.biu.ac.il} \and Bar-Ilan University, \email{orenkapah.ac@gmail.com, porately@cs.biu.ac.il} \and University of Michigan, \email{kopelot@gmail.com} \and Weizmann Institute of Science , \email{moni.naor@weizmann.ac.il}   }
%
\author{
Amihood Amir\thanks{Bar-Ilan University and Johns Hopkins University,
  Email:amir@cs.biu.ac.il. Partly
  supported by  NSF grant CCR-09-04581 and ISF grant 347/09.}\and
Oren Kapah\thanks{Bar-Ilan University,Email:orenkapah.ac@gmail.com}\and
Tsvi Kopelowitz \thanks{University of Michigan, Email: kopelot@gmail.com}\and
Moni Naor\thanks{Weizmann Institute of Science, Email: moni.naor@weizmann.ac.il. Incumbent of the Judith
        Kleeman Professorial Chair. Research supported in part by  grants from the
        Israel Science Foundation and from the
I-CORE Program of the Planning and Budgeting Committee and the Israel Science Foundation, BSF and IMOS.}\and
Ely Porat\thanks{Bar-Ilan University Email: porat@cs.biu.ac.il} }

\date{}

\maketitle

%
%\author{Amihood Amir\inst{1} \fnmsep \inst{2} \fnmsep \thanks{Partly
%    supported by  NSF grant CCR-09-04581 and ISF grant 347/09.} \and
%  Oren Kapah\inst{1} \and Tsvi Kopelowitz \inst{3} \and Moni Naor
%  \inst{3} \and Ely Porat \inst{1}}
%
%\institute{Department of Computer Science, Bar-Ilan University,
%Ramat-Gan 52900, Israel. \email{E-mail: amir@cs.biu.ac.il} \and
%Department of Computer Science, Johns Hopkins University,
%Baltimore, MD 21218. \and Weizmann Institute of Science, Rehovot,
%Israel. \email{E-mail: kopelot@cs.biu.ac.il} }
%\parindent 0pt
%\parskip .1in
%\maketitle
\thispagestyle{empty}
\setcounter{page}{0}

\begin{abstract}

We introduce and examine the {\em Holiday Gathering Problem} which
%MN2 removed
%has an interesting social twist as it
models the difficulty that couples have when trying to decide with which parents should they spend the holiday.  Our goal is to
schedule the family gatherings so that the parents that will be {\em happy}, i.e.\ all
their children will be home {\em simultaneously} for the holiday festivities, while minimizing the number of consecutive holidays in which parents are not happy.

The holiday gathering problem is closely related to several classical problems in computer science, such as the {\em dining philosophers problem} on a general graph and periodic scheduling,and has applications in scheduling of transmissions made by cellular radios. We also show interesting connections between periodic scheduling, coloring, and universal prefix free encodings.
%, where on every edge lies a resource (for example, the children). A philosopher %can eat only when all the resources
%are available, which in the holiday gathering settings translates to parents %being happy.

%MN3
The combinatorial definition of the Holiday Gathering Problem is: given a graph $G$, find an infinite sequence of independent-sets of $G$. The objective function is to minimize, for every node $v$, the maximal gap between two appearances of $v$.
In good solutions this gap depends on local properties of the node (i.e., its degree) and the the solution should be periodic, i.e.\ a node appears every fixed number of periods. We show a coloring-based construction where the period of each node colored with the $c$ is at most $2^{1+\log^*c}\cdot\prod_{i=0}^{\log^*c} \log^{(i)}c$ (where
$\log^{(i)}$ means iterating the $\log$ function $i$ times). This is achieved via a connection with {\it prefix-free encodings}. We prove that this is the best possible for coloring-based solutions. We also show a construction with period at most $2d$ for a node of degree $d$.

%(2) A `heavyweight' non-periodic simple solution that guarantees that the parents of $d$
%children will host all of them at least once in every $d+1$
%years. (3) A lightweight perfectly-periodic
%solution  in which the parents of $d$ children are
%guaranteed to host them all together at least once in every $2^{\lceil\log d \rceil}\leq 2d$
%years. (4) A lightweight coloring-based technique in which parents host their children every set number of years, and for parents with color $c$ the number of holidays
%between being happy is at most $2^{1+\log^*c}\cdot\prod_{i=0}^{\log^*c} \log^{(i)}c$ where
%$\log^{(i)}$ means iterating the $\log$ function $i$ times. This is
%achieved via a connection with {\it prefix-free encodings}. We also prove that
%the performance of this algorithm almost matches a lower bound for coloring-based
%algorithms of this scheme. Finally, (5) we discuss the dynamic
%situation, where couples of children can be formed or dissolved - which is
%modeled by a graph with fixed nodes (parents) and edges (couples) that can be added or
%deleted.

\end{abstract}

\setcounter{page}{0}
\newpage

\section{Introduction}\label{s:int}

%In an American Psychological Association study, $44\%$ of women and
%$31\%$ of men reported heightened stress during the holiday
%season~\cite{apa:06}. Part of that stress was due to ``family
%reasons''. In every culture family holiday gatherings play an
%important role (see e.g.~\cite{nrg:09,tnn:10}).

In every culture family holiday gatherings play an
important role (see e.g.~\cite{nrg:09,tnn:10}), but these gatherings are also stressful.
In this paper we consider one of the anxiety-causing problems before
the holidays - where to go for the holiday dinner? Parents, whose
children are in a monogamous relationship, would obviously (?!) like
to have {\em all} their children at home for the holiday meal
(i.e.\ there is a special pleasure gained by the festive experience of
hosting all the children simultaneously and intuitively the goal is to
have this event occur as frequently as possible). We say that such
parents wish to be \emph{happy} during the holiday\footnote{Another possible goal would be to assure that no parents
are left alone for the holiday (i.e.\ there is a special depression
suffered by not hosting any of the children and intuitively the goal
is to have this event occur as infrequently as possible), which we discuss in the Appendix.}. However, the
conflict is that the in-laws would also be happy if all their children
come to them.

The astute reader realizes by now that, being computer
scientists, we are not really equipped with tools for coping with the
psychological and social problems involved. However, it turns out that
these problems are attractive from a computer science point of view
for reasons other than the social one. First, it is closely related to
the classic {\em dining philosophers problem}, which we recall in
section~\ref{ss:related_work}. Furthermore, scheduling problems are
part of our mainstay, and these problems are indeed scheduling
problems. To this end, we focus on algorithms whose goal is to
schedule which parents are happy during any given holiday in a {\em distributed} manner with the
objective of minimizing the number of consecutive holidays in which a
parent is not happy\footnote{One may consider the problem of
maximizing happiness for a given year, but it is straightforward to
see that this problem is $\mathbb{NP}$-hard. One may also consider an
objective function relating to the satisfaction of parents. Details
for both can be seen in the Appendix.}.

The holiday gathering problem has direct applications in the realm of common resource scheduling. Suppose that in a world with many agents, each agent requires some shared resources in order to perform some job. For example, it would be beneficial if cellular radios could guarantee that when they broadcast none of the other radios interfere. In this application the shared resource is the air which is within transmission radius of more than one radio. We can model this as radios being parents and two radios which share some air are modelled as in-laws.

%MN added this
\paragraph{Connection to coloring:} As in many problems in computer science, for some special inputs the
problem is simple. For example, imagine a society partitioned into two
groups, say $A$ and $B$, where only intergroup marriage is allowed and
the goal of the parents is to be happy. In this case there is a very good
arrangement: for the first holiday members of group $A$ host all of
their children, and from then onwards groups $A$ and $B$ alternate on
who hosts for each holiday. Thus, every two years a family can gather
all its children for a holiday dinner, no matter how many children it
has. Why did it work out so well?
%MN Not clear that we shouldn't define it this way, rather than as
%bipartite graph.
Consider the conflict graph: nodes are the families and there is an
edge between the corresponding nodes if a child of one family married
a child of the other. The hosting families of each holiday constitute
an independent set in the graph. We would like to cover the graph with
as few independent sets as possible - a coloring problem. In the above
example we had a bipartite graph. However, life is usually not that
simple, and in fact marriage is rarely arranged so as to create a
bipartite graph.

%Ami
In a general graph where $\Delta$ is the largest degree of any node in
the graph, it is immediate that one can color the graph in $\Delta+1$
colors and parents are happy on their color in every cycle of
$\Delta+1$ holidays. This gives a guarantee of happiness every $\Delta+1$
holidays. However, this solution is not pleasing. There is an
uncomfortable feeling in making the parents of a single child wait
$\Delta+1$ holidays between happy holidays simply because some other parents
have a large brood. We would like the bound on the distance
between happy holidays of a parent to be dependent on {\em local} properties of the parent, like their degree or color,
rather than {\em global} properties of $G$ such as the the maximum degree in the graph. Indeed, we present some
algorithms where the bound on the maximum distance between happy holidays
for every parent is dependent on the parent's degree and not $\Delta$.

The example of the intergroup marriages above, demonstrated that, in
some cases, the bound on the maximum distance between happy holidays
can be much better than the degree of the parent node. The better
result was achieved by the fact that there was a bipartite coloring of
the graph. This observation motivates strengthening the connection
between scheduling the holidays for which each parent
is happy and the chromatic number of the graph.

%MN
Suppose that we are looking for a schedule that will minimize the
maximum time any parent has to wait until it is happy. Then this
problem is as hard as coloring the conflict graph $G$ which can be
seen as follows. Suppose that there is a schedule where no parent has
to wait more than $c$ years to be happy. Then we can construct a legal
coloring for $G$ by observing $c$ consecutive holidays. The set of
happy parents in a given holiday form an independent set. It is
straightforward to see that $G$ can be partitioned into $c$
independent sets and each such set can be colored by its own color. On
the other hand, if there is a coloring with $c$ colors (thought of as
values in $\{1, 2, \ldots, c\}$), then there is also a schedule that
makes every parent happy in $c$ steps: on year $i$, parents whose color
is equal to $(i \mod c)+1$ are happy.

\paragraph{What is the fair share of parents?} Defining fairness is the subject of much debate in philosophy, game theory and theology. Much of cooperative game theory deals with fair allocation of resources. In our case the problem seems hard: given the tight relationship with coloring and maximum independent set, we cannot even determine efficiently the `value' of the full coalition (see Appendix~\ref{sec:hard_fair}). On the other hand, consider the following simple `chaotic' process called ``first come first grab": parents wake up at a random time and grab their available (those who have not been grabbed) children. The probability that a node $p$ manages to grab {\em all} its children is $1/(deg(p)+1)$. So the expected time until hosting all the children is $deg(p)$ and this is the landmark we will try to obtain, i.e.\ we would like every parent to host a holiday with all their children every $O(deg(p))$ years. It is also clear that in general we cannot hope to get a better than $deg(p)+1$ result, if the conflict graph is a clique.

\subsection{Our Goals}
Generally speaking, we want a scheduling which will determine which
parents will be happy in any given year to have the following
properties:
\begin{itemize}
 \item {\bf A local-bound:} The frequency in which parents are happy
   should be a bounded function of some local properties of the
   parent such as the number of children, the size of the local
   neighborhood, or the parent's color if we can a-priori color the
   conflict graph. The bound should not depend on a function of global
   graph properties such as the maximum degree in the graph or the
   total graph size.

\item {\bf Lightweight:} The hosting schedule should be easy to
  determine from a small amount of local data, and with a small amount
  of communication between the parents. For example, we would like
  parents to know in advance the years in which they will be happy
  from a short piece of information (like the parent's color).

\item {\bf Periodic:} From a long term planning point of view, it is desirable to have a periodic scheduling guaranteeing that a parent always waits the same number of holidays in between happy holidays.
    %MN3 added
    In the context of scheduling radio transmissions the advantage of a periodic solution is that a node does not have to waste energy between periods where it can transmit.
    % is that the Achieving such a scheduling is also interesting from a combinatorial aspect.

\item {\bf Distributed:} The holiday gathering problem is distributed by nature; a parent can be seen as processes in a very large network trying to achieve some common scheduling goal. One would like to have algorithms that work well in a distributed setting, as otherwise the communication burden on a parent may be too large.

\end{itemize}

%MN3
To summarize, the combinatorial definition of the Holiday Gathering Problem is: given a graph $G$, find an infinite sequence of independent-sets of $G$. The objective function is to minimize, for every node $v$, the maximal gap between two appearances of $v$.
In good solutions this gap depends on local properties of the node (i.e., its degree) and the the solution should be {\em periodic}, i.e.\ a node appears every fixed number of periods.

\subsection{Our Results and Techniques}
The main contributions of this paper are, thus:
\begin{enumerate}
\item Providing a combinatorial definition for the {\em Family Holiday
 Gathering Problem}.
\item Providing a \emph{non-periodic} solution for the
  problem. In this solution we guarantee that a parent of $d$ children
  will be happy at least once in every $d+1$ years. We remark that although this solution is simpler from a technical prospective, it is the best guarantee that is achievable even for heavyweight solutions, as a function of the degree only (i.e.\ for some graphs better solutions exist) and hence it provides some idea of what is a natural lower-bound for lightweight solutions.

 The downside of this algorithm is that it is heavyweight:  it either uses extensive
 communication after each holiday, or requires a large  amount of
 local memory per node (for instance, the full topology of the graph) to determine the schedule of that parent. Furthermore, this solution is a-periodic,
 i.e.\ the number of years that pass between two holidays in which the
 parent is happy is not set;
 %MN3 added
 as far as we know, the length of period in which the schedule repeats may be exponential in the size of the graph.
 %MN
 %This algorithm is {\em global} and assumes a {\em static} dataset.
 % The algorithm utilizes a {\em phased} greedy graph mapping algorithm.
%\item Provide a dynamic solution, since the status of families changes
% annually.

\item Providing two {\em lightweight perfectly-periodic} solutions:
%Unlike the medical residency
% match, it is unlikely that a central solution can be worked out by
% the government, thus we explore the possibilities of local
% strategies that optimize satisfaction.
\begin{enumerate}
\item
{\bf Color-bound:} The first algorithm is based on any coloring of the
graph. The number of years a parent with color $c$ has to wait to be happy
is at most $2^{1+\log^*c}\cdot\prod_{i=0}^{\log^*c} \log^{(i)}c$ where $log^{(i)} c$ is the iterative log function of $c$ takin $i$ times. This is achieved via an interesting
new technique via a connection with {\em prefix-free encoding}. We show that, for color-based techniques, our algorithm is {close to}
optimal. This is done by proving a lower bound of $\prod_{i=0}^{\log^*c} \log^{(i)}c$ for scheduling algorithms based on
graph coloring using the Cauchy condensation test.
\item
{\bf Degree-bound:} The second algorithm requires a special type of
coloring that can be obtained via a greedy algorithm. In this
solution, a parent of $d$ children will have all its children at home
for the holidays exactly every $2^{\lceil \log d \rceil} \leq 2d$
years.
 \end{enumerate}
\end{enumerate}

%We also prove a better than $d$ lower bound for periodic algorithms that are envy free: all node with the same neighborhood have the same period. This is described in Appendix~\ref{s:lower_bound_periodic}.
We emphasize that we are mainly interested in periodic scheduling, but we also discuss some non-periodic solutions that serve as a sanity check and help us understand what is the best one could hope for without any constraints. 

\subsection{Related Work}\label{ss:related_work}

As mentioned, the holiday gathering problem is related to several lines of investigation in computer science (not to mention other scholastic activities).
Issues related to calendrical calculations have attracted the best minds since antiquity (see Dershowitz and Reingold~\cite{DershowitzN07}). No lesser than al-Khwarizmi (after whom the term `algorithm' is named) wrote a treatise on the Hebrew Calendar (``Risala fi istikhraj ta'arikh al-yahud"), see Knuth~\cite{Knuth79}. Some calendars are fixed in the sense the that it is known in advance when each holiday will occur, as is the case with the Western (Gregorian) calendar and the Hebrew calendar, while others, like the Muslim Calendar or the Old Hebrew Calendar are determined on-the-fly, e.g.\  based on lunar observations. This is reminiscent of some of the issues that arise in our algorithms (the one in Section~\ref{s:global} vs.\ those of Sections~\ref{s:local} and~\ref{s:degree_bound}).

The dining philosophers problem is the famous  resource allocation problem introduced by Dijkstra~\cite{Dijkstra71}, see Lynch~\cite{Lynch96}. In this problem,
there is a given conflict graph where each node represents a processor and each edge represents a resource (a ``fork" in the story where the processors are philosophers who would like to eat) which is shared by the two
endpoint processors.  At any time, a fork can be ``owned" by
at most one of the processors that share it. Each processor can be in one of three
states: resting, hungry, or eating. %The processors operate asynchronously.
A resting processor can become hungry at any time. In order to eat, a processor must obtain all the forks on its adjacent edges.
A processor eats for at most a bounded time, after which it returns to the
resting state. The problem, and its many variations, have played a major role in
concurrent programming and distributed computation.
For this problem there are solutions that minimize the wait chain, based on coloring of the edge or the nodes (see \cite{Lynch81, StyerP88, ChoyS95, NaorS95, MayerNS95}). The main difference between the focus of this work and most work on the dining philosophers problem is that we  assume that the  philosophers want to eat {\em all} the time (i.e.\ they become hungry right after they finish eating)\footnote{So one may call the problem ``The Fressing Philosophers Problem", as suggested by Cynthia Dwork.} and that the meal takes a fixed amount of time and the main issue is how can we provide an efficient and high throughput solution  while guaranteeing some  reasonably fair allocation.

Problems related to ours have appeared frequently in the scheduling literature, starting perhaps from the Chairman Assignment Problem of Tijdeman~\cite{Tijdeman1980}. The assumption is that there is one resource and all users are interested in using it. Each user or task has a weight and the goal is to schedule the users so that they obtain the resource proportionally to their weight.  Sometime there are multiple identical resource, but each task can be assigned to one resource concurrently~\cite{BaruahCPV93,LitmanMS11}. This is similar to our setting when the graph is a clique (or composed of components that are cliques) and all the weights are the same.
In a `perfectly periodic scheduling' the goal is to schedule the users in a {\em periodic way} (every user $i$ gets the resource every $\tau_i$ rounds)  where each user gets it weight (see \cite{Bar-NoyNP02}). The algorithms of Sections \ref{s:local} and~\ref{s:degree_bound} are perfectly periodic.
% BrakerskiNPS02,
In the chromatic sum problem (see~\cite{Bar-NoyBHST98}) one tries to find a coloring that minimizes the the total sum of the colors (where the assumption is that the colors are in $\mathbb{N}$). The motivation is, again, from scheduling with the goal of minimizing {\em average} waiting time, rather than the maximum waiting time as a function of the degree, as is our goal.

Fairness in online scheduling has received some attention as well, for instance the carpool problem~\cite{FaginW83,AjtaiANRSW98,Naor2005}, that can also be viewed as a generalization of the Chairman Assignment Problem of Tijdeman~\cite{Tijdeman1980}.

The ${\cal LOCAL}$ model of computation in distributed computing was first considered by Linial~\cite{Linial92} and much developed since then, see  Peleg \cite{Peleg00}. The problems of interest  are especially those of coloring and maximal independent set. For both of these problem good randomized algorithms are known, see the monograph by Barenboim and Elkin for a survey of recent results~\cite{BarenboimE13}. Coming up with deterministic polylog in $n$ algorithms for these problems is a major open problem in the area.

\section{Preliminaries}\label{s:p}
%In this section we give the formal definitions of our model and the
%concepts of {\em happiness}.

Our universe is a conflict graph $G=(P, E)$ where $|P|=n$. One may view our model as consisting of a set of {\em
parents} $P$, and a set of conflict {\em children} edges $E$, i.e.\ children that are in a relationship.
For any $p\in P$ denote by $E_p\subseteq E$ the set of edges that touch $p$ in $G$.
Note that singular children,  married children that do not have both parents, siblings that marry each other  only simplify things.
Most of this work assumes that the conflict graph is fixed, but we also consider the dynamic case and discuss what can be done.

\begin{definition}\label{d:gathering}
A {\em family holiday gathering} (a {\em gathering},
for short) is an orientation $h$ of the edges in $E$ where each $e\in E$ is assigned a direction. We say that node $p\in P$ is {\em happy} in $h$ if $p$ are a sink (all edges incident on $p$ are directed toward $p$). The set of happy nodes in a given orientation is an independent set.
\end{definition}

%\begin{definition}\label{d:gathering}
%A {\em family holiday gathering} for holiday $h$ (a {\em gathering},
%for short) is a partition $\gamma = \{S_i\}_i$ of $E$, where for each
%$i$ there exists ${p_i\in P}$ such that for all $e\in S_i$ we have $e\in E_{p_i}$. We
%can therefore denote $S_i$ by $S_{p_i}$.
%\end{definition}

%\begin{definition}\label{d:happiness}
%We say that parents $p\in P$ are {\em happy in a family holiday
% gathering} $\{S_i\}_i$ if for all $e\in E_p$ we have $e\in
%S_p$. In words -- all children of $p$ come home for the holiday.
%
%\end{definition}

%We now give a formal definition of our goals.
%MN3 did some mowing
\begin{definition}\label{d:seg}
Consider a sequence $H=h_i,\,\ i=1,...,\infty$ of
gatherings and denote the subsequence $h_i,...,h_j$ by
$h[i:j]$. % and call it an {\em interval} of $H$.
%For $p\in P$.
If $p \in P$ is not happy at any $h \in h[i:j]$ then we
call the interval $h[i:j]$ an {\em unhappiness interval} for
$p$ and call the longest such interval a {\em maximum unhappiness interval} of $p$ and denote its length by $mul_H(p)$ (or $mul(p)$ when $H$ is clear from the context).
% if there is
%no unhappiness interval of $H$ for $p$ whose length is
%greater. % than $j-i+1$ (the length of $h[i:j]$).
%We call the length of a maximum unhappiness interval of $p$ the {\em maximum unhappiness
% length} of $p$, and denote it by $mul_H(p)$ or $mul(p)$ when $H$ is clear from the context. % we will abuse notation and simply denote .
\end{definition}

Intuitively, our goal is to minimize the maximum unhappiness
length. However, as discussed above, we would like $mul_H(p)$ to be bound by some function of properties that are local to $p$ and not dependent on global graph properties. In a {\em periodic solution} we would like every node $v$ to be happy every fixed number of periods.

\section{The Non-Periodic Degree-bound Algorithm}\label{s:global}

%MN moved up; the satisfaction problem is not all that interesting.
%While achieving maximum satisfaction is fast, the solution is not
%socially acceptable, since the same parents will be satisfied every
%year while others will never be satisfied.
%What is necessary is some
%scheme whereby all parents are guaranteed satisfaction within some
%cycle of time.

%MN modified
We now start dealing with our main goal, guaranteeing {\em happiness} to
every node within a reasonable cycle of years. We present\footnote{We remark here that while this algorithm is less interesting from a technical perspective, it is useful for set up, and gives us a benchmark for comparison when attempting to understand the strengths of the lightweight algorithms.
} an algorithm guaranteeing that a node of degree $d$ children has to wait at most $d+1$ steps till  it is happy.  However, the node does not know in advance all the times in which it will host the holiday, just the next time it will  do so.

The algorithm we consider starts with a {\em distributed graph
coloring algorithm}. As mentioned in Section~\ref{s:int}, there is a simple mechanism using a $\Delta +1$ coloring to obtain a
sequence $H$ such that for every node $p\in P$ we have $mul_H(p) = \Delta +1$. Such a coloring can be obtained in a distributed manner by applying, for example, the recent randomized algorithm of Barenboim, Elkin, Pettie, and Schneider~\cite{BEPS12} (denoted by the BEPS algorithm for short) running in $O(\log \Delta + 2^{O(\sqrt{\log \log n})})$ rounds. The BEPS algorithm also has the property that the color $c$ picked for a node with degree $d$ will always be bound by $c\leq d+1$ (see~\cite{Johansson99}, which is used as a black box in~\cite{BEPS12}, for details). So at first the smaller degree nodes will be happy pretty quickly. However, for their next turn they will have to wait time proportional to $\Delta$.
As mentioned, we seek an efficient mechanism for constructing a gathering sequence $H$ where $mul_H(p)$ depends on local properties of $p$.

%MN added

To solve this problem we will use a phased algorithm
where colors are reassigned every phase (holiday), providing a sequence of
gatherings. The initial coloring is the one obtained by the BEPS algorithm.

At holiday $i$, greedily,
%MN
%starting from the vertices with the
%smallest degree,
%MN This is only relevant in the first phase - afterwards
%whose colors are less than $i$. Color each such
recolor the nodes whose current color is $i$: color each such
node $v$ with the smallest number $j>i$ such that none of $v$'s neighbors
has color $j$. At holiday $i$ the nodes colored $i$ are happy. % -- all their children are at home for the holiday.

The algorithm appears in detail below. We denote by $col(p)\in
\mathbb{N}$ the current color of node $p$.

%MN modified
\fbox{
\begin{minipage}{15cm}
{\bf Algorithm -- Phased Greedy Coloring}
{\sf
\begin{enumerate}
   \item {\bf Initialization:} Assign every node $p$ a color.

   \item For $i=1$ to $\infty$
\begin{enumerate}
    \item For every node $p\in P$: if $col(p)=i$ then make $p$ happy and recolor $p$:
\begin{enumerate}
\item Let $p_1,...,p_{{\ell}}$ be the nodes
  adjacent to $p$ in $P$.
\begin{enumerate}
\item Let $s=min \{ t | i<t \leq i+{\ell}+1,\ \  t \not\in \{  col(p_{1}),...,c(p_{{\ell}})\}\}$.
\item $col(p)\leftarrow s$.
\end{enumerate}

\end{enumerate}
\end{enumerate}
\end{enumerate}

{\bf end Algorithm}
}
\end{minipage}
}

\vskip .1in

For every phase $i$ we perform $O(1)$ rounds of communication, %\log\log d$,
%where $d$ is the degree of the graph.
This is true since every node needs to only communicate with its neighbors and choose the smallest number
that is both greater than $i$ and different from the color of all its
neighbors.
%If this is done via, e.g. the van Emde Boas~\cite{veb-77}
%data structure, the time is $O(\log\log d)$ per check in step
%2.(a)i.A.

\begin{theorem}\label{t:color2}
There is a holiday scheduling algorithm, whose initialization takes $O(\log \Delta + 2^{O(\sqrt{\log \log n})})$ rounds,
and executing each holiday takes another $O(1)$ rounds, which guarantees
that $\forall p\in P$ we have $mul_S(p)\leq d_p+1$, where $d_p$ is the degree of $p$. In words, for
every node $p$, within every sequence of $d_p+1$ holidays
$p$ are happy at least once.
\end{theorem}

\begin{proof} Apply the Phased Greedy Coloring algorithm. The number of rounds was
shown to be $O(\log \Delta + 2^{O(\sqrt{\log \log n})})$ for the initialization and $O(1)$ every phase.
For every node $p\in P$, if it is made happy at phase $i$ then in
phase $i+1$ it is re-colored. The number it chooses is the smallest
number that exceeds $i$ and is not equal to the number of any of its
neighbors. However, since it has $d_p$ neighbors then the color it
gets can not exceed $i+d_p+1$. \end{proof}

%MN3
Recall that if we think of the ``first come first grab" where nodes wake up at random at grab all there available neighbors the probability of happiness of a node is $1/(d+1)$, so the expected time till happiness is $d+1$. This algorithm can be seen as providing a guarantee for the waiting time.

Notice that this algorithm generally does {\em not} give a periodic schedule. Also, it requires communication to take place at every phase, which implies the need to invest a lot of energy if we are considering an application such as cellular communication. An alternative solution would be to have each node know and remember {\em all} of the conflict graph locally (and then simulate the algorithm locally to obtain determine when it should host). But such an approach requires a lot of local memory, may be unfeasible, and may cause privacy issues.

\section{A Periodic Lightweight Color-bound Algorithm}\label{s:local}

We are now interested in providing a mechanism whereby the decision of
when a node is happy is dependent on {\em local} properties of that node. We present a general
scheme for such a mechanism that depends on the color nodes receive during an initial coloring algorithm. Once again we begin by distributively coloring the graph. However, this time we do not make any assumptions on the coloring algorithm, and so this algorithm works for any graph coloring, including the (possibly difficult to obtain) optimal one.

%In general, our algorithm assumes an infinite sequence of holidays $i=1,..., \infty$.
We consider a mapping of the holiday numbers to colors. The mapping needs to satisfy two conditions: (1) Every holiday number is mapped to at most one color, and (2) The resulting gathering sequence guarantees that nodes do not have very large maximum unhappiness intervals.

\fbox{
\begin{minipage}{15cm}
{\bf Algorithm Scheme}
{\sf
\begin{enumerate}
    \item Color $G$.
    \item At holiday $i$, if $decode(i)=col(p)$ then make $p$ happy.
\end{enumerate}

{\bf end Algorithm}
}
\end{minipage}
}
\vskip .1in

{\bf Example:}
\begin{enumerate}
\item {\bf Trivial:} Consider the trivial example where the nodes are
colored sequentially from $1$ to $|P|$. At holiday $i$, make $p=i {\ \rm
mod\ } |P|$ happy. No two adjacent nodes are encoded to the same
number, but $\forall p\in P$ we have $ mul(p)=|P|$, which means that it depends on global properties of $G$.
\item {\bf Prefix Free Color Code:} Apply on $G$ the BEPS distributed graph coloring algorithm that colors each graph node $p$ by a color not exceeding $deg(p)+1$.
Now encode the colors using some prefix-free binary code. On
holiday $i$, consider the binary representation of $i$ from right to
left (with an infinite sequence of $0$'s padded to it). Any node $p$ is made happy if the prefix-free encoding of $col(p)$ is a
prefix of $i$.

The solution is appropriate, since no two adjacent nodes will be
made happy concurrently at any given holiday $i$: they will be assigned different colors and the
two different colors are encoded in a prefix-free code and hence the binary representation of $i$ cannot encode both of them. For every $p\in P$, we have that $mul(p)$ of this
procedure is dependent on the length of the prefix-free encoding of $p$'s color.
\end{enumerate}

We will indeed use prefix-free binary codes in our algorithm in Section~\ref{ss:elias}. Notice that if a node is given a color $c$ which uses $x_c$ bits in its prefix-free code, then the schedule of when that node is happy is periodic with period $2^{x_c}$. In other words, every $2^{x_c}$ holidays the node will be happy

\subsection{Lower Bounds.}\label{ss:lower}

The coloring-based algorithm scheme of Section~\ref{s:local} starts by
assigning colors to the nodes in $G$. We would like every node to use
its color in order to compute a period $\pi$ such that every $\pi$ years that node
is happy. In this section we compute lower bounds on that period as a
function of the color. This defines the best period one can hope to
achieve by the coloring-based scheme. In Section~\ref{ss:elias} we
show an algorithm which guarantees a period, that almost matches the
lower bound.

We recursively define the function $\phi:\mathbb{N}
\rightarrow \mathbb{R}$.

\begin{definition}\label{d:phi}
$$\phi(i)=
\begin{cases}
1 & {\rm if\ } i\le 1,
\cr
i\cdot \phi(\log i) & {\rm if\ } i>1.
\end{cases}
.$$
Explicitly, $\phi(i)= i\cdot \log i\cdot \log \log i\cdot \log \log
\log i\cdots 1$, or $\phi(i)=\prod_{i=0}^{\log^*c} \log^{(i)}c$, where
$\log^{(i)}$ means iterating the $\log$ function $i$ times.
\end{definition}

\begin{theorem}\label{t:lowerbound}
Let $G$ be legally colored  and every node $p\in P$ has a color $c_p$. Let sequence $H=h_i,\,\ i=1,...,\infty$
of gatherings be such that there is no gathering $h_{i_0}$ in
which two different colors, $c_1 \neq c_2$ are happy. If there exists a function $f:\mathbb{N} \rightarrow \mathbb{N}$ such that for every $p\in P$ we have $mul(p)=f(c_p)$, then $f(c)\in \Omega(\phi(c))$.
\end{theorem}

\begin{proof}
Consider a subsequence $h[i_0:j_0]$ of length $m$.
We require that at any gathering in the sequence, there is one
and only one color $c$ which is happy. The conclusion
is that there are at least $\lceil {m \over f(c)} \rceil$ occurrences in
the sequence where nodes of color $c$ are happy. This is true for all
colors $c$ that get mapped to the subsequence. Let $C = \{ c | c {\rm
  \ is\ a\ color\ mapped\ to\ }h[i_0:j_0]\}$. Then we have
$\sum_{c\in C} \lceil {m\over f(c)} \rceil \leq m\Rightarrow \sum_{c\in C} {1\over f(c)} \leq 1.$
Taking into account all possible sets $C$ as the size of $G$ goes to infinity, we have
$\sum_{c=1}^{\infty} {1\over f(c)} \leq 1.$
Clearly this will not hold if we have $f(c)=c$. It
holds when $f(c)=2^c$ and even when $f(c)=c^{1+\epsilon}$ for any constant $\epsilon >0$. According
to Cauchy's condensation test~\cite{bkk:06} the smallest function for
which this inequality holds is $f(c) = {\phi(c)}$, i.e.\
$f(c)=\prod_{i=0}^{\log^*c} \log^{(i)}c.$
%\qed
\end{proof}

\subsection{Upper Bound - Elias Code}\label{ss:elias}

Theorem~\ref{t:lowerbound} proves that one can not hope, in any
color-based scheme, to achieve an assignment of colors to the nodes where for every $p\in P$ with color $c$ we have $mul(p) = o(\phi(c))$.
In this section we will almost match this lower bound. We guarantee that no two nodes with
different colors are made happy during the same holiday, and that $mul(p)=
\phi(c) 2^{\log^* c+1}$. This algorithm is based on the {\em Elias omega
  code}. The Elias omega code is a universal code for the natural numbers
developed by Peter Elias~\cite{Elias75}. It is one of the prefix-free
codes that represents the integers by a number of bits relative to their
size. While the omega code is not the most practical code,
it is theoretically the most efficient Elias code. Our algorithm below
is correct for all Elias codes, but we chose the omega code for its
almost optimal complexity.

%<<<<<<< .mine
%\begin{definition}\label{d:elias}
%Let $n\in \mathbb{N}$. Denote by $B(n)$ the binary representation of
%$n$, i.e.\ $B(n)$ is a string over $\{ 0,1\}$ which is the
%representation of $n$ in base $2$ with no leading zeros.
%Denote the number of bits in $B(n)$ (the highest power of $2$, $b$
%such that $2^b \leq n$) by $|B(n)|$. Let $S$ be a binary
%number. Denote by  $LSB(S,k)$ the suffix of $S$ of length $k$, or the
%$k$ least significant bits of $S$.
%=======
Details of the Elias omega code are given in the Appendix. The important properties that we need from the code are that it is a prefix free code, and that the length of the coding of $i$ is given by $$\rho(i)= 1+ \lceil \log(i) \rceil + \lceil
  \log(\lceil \log(i) \rceil -1) \rceil + \lceil \log(\lceil
  \log(\lceil \log(i) \rceil -1) \rceil - 1) \rceil + \cdots.$$
%>>>>>>> .r44

Our algorithm will use the Elias omega code of the colors in
reverse. Denote by $S^R$ the string $S$ reversed from left to
right. For example, $(abcdef)^R = fedcba$. Denote by  $LSB(S,k)$ the suffix of $S$ of length $k$, or the $k$ least significant bits of $S$. Denote by $B(n)$ the binary representation of $n$, i.e.\ $B(n)$ is a string over alphabet $\{ 0,1\}$ which is the
representation of $n$ in base $2$ with no leading zeros.
\vskip .1in
\fbox{
\begin{minipage}{15cm}
{\bf Elias omega code Algorithm}
{\sf
\begin{enumerate}
    \item Color $G$.
    \item At holiday $i$, if $LSB(B(i))=\omega (p)^R$ then make $p$ happy.
\end{enumerate}

{\bf end Algorithm}
}
\end{minipage}
}

\vskip .1in

\begin{theorem}\label{t:eomega}
The Elias omega code algorithm guarantees happiness for node
$p\in P$ in every cycle of length bounded above by $\phi(c) 2^{\log^*c
  +1}$, where
$c$ is the color of $p$.
\end{theorem}

\begin{proof}
The time between happy holidays for $p$ is $2^{\rho(c)}$. Notice that

\begin{align*}
\rho(c)
& = 1+ \lceil \log(c) \rceil + \lceil
  \log(\lceil \log(c) \rceil -1) \rceil + \lceil \log(\lceil
  \log(\lceil \log(c) \rceil -1) \rceil - 1) \rceil + \cdots\\
& \leq 1+\sum_{i=1}^{\log^*c} \lceil log^{(i)}c\rceil \leq 1+\log^*c + \sum_{i=1}^{\log^*c} log^{(i)}c.
\end{align*}

The ${\log^*c}$ term is a result of
the rounding up of the $\log$ at every recursion of $\rho(n)$, since
the number of bits can not be a fraction. For some colors this term may be less than $\log^*c$, or even a constant, but it will
never exceed $\log^* c$. Now, the time between happy holidays is
\begin{align*}
2^{\rho(c)}
& \leq 2^{1+\log^*c + \sum_{i=1}^{\log^*c} log^{(i)}c} = 2^{1+\log^*c } \cdot \prod_{i=0}^{\log^*c-1} log^{(i)}c\\
& =2^{1+\log^*c } \cdot \prod_{i=0}^{\log^*c} log^{(i)}c =2^{1+\log^*c } \phi(c).
\end{align*}
\end{proof}

\section{A Periodic Lightweight Degree-bound Algorithm}\label{s:degree_bound}

While many graphs have a low chromatic number that may be obtainable\footnote{For instance, for triangle-free graphs Pettie and Su~\cite{PettieS13} very recently gave an $O(\Delta/\log \Delta)$ distributed coloring.}, other graphs do not. Furthermore, while some classes of graphs may have a low chromatic number, it is not clear how to obtain algorithms that achieve $o(\Delta)$ colors for some of these classes. In such cases a color-bound algorithm may not suffice considering the inherent lower-bound. We could use the bound of $c\leq d+1$ for a node with color $c$ and degree $d$, to obtain a guaranteed degree-bound of $mul(p) \leq \phi(d+1) 2^{\log^*(d+1) +1}$ using Theorem~\ref{t:eomega}, but we can do better.

To this end, we present a second local algorithm which is degree-bound, i.e.\ the length of the cycle for which a node $p$ wait is bounded by a function of their degree $d$. In particular, our algorithm guarantees that this bound is at most $2d$, which is very close to the bound $d+1$ obtained by the non-periodic algorithm of Section~\ref{s:global}. We first describe a sequential version of our algorithm, and later show how to convert it to the distributed setting.

\subsection{A Sequential Algorithm}
The periodic scheduling is obtained using a greedy algorithm, where nodes are arranged in {\em decreasing} order of their degrees and each node $p$ in its turn chooses a non-negative integer. Thus, when $p$ with degree $d$ has to chose an integer, none of its smaller than degree $d$ neighbors has chosen an integer. Let $j=\lceil \log (d+1) \rceil$. When it's $p$'s turn to pick an integer, there must be at least one integer $x$ in the range $[0,2^{j}-1]$ such that no neighbor of $p$ has chosen a number $x'$ where $x=x' \text{ mod } 2^j$. So $p$ picks $x$ to be its integer.

The hosting assignments for holiday $t$ now is determined by checking whether $t\equiv x \mod 2^j$. If it 'yes' then $p$ is happy in holiday $t$, otherwise it is not. We can immediately see that $t$ is happy once every $2^j = 2^{\lceil \log (d+1) \rceil} \leq 2d$ years. We will now show that there are no conflicts.

\begin{lemma}\label{lemma:sequential_degree_bound}
Let $p_1$ and $p_2$ be two adjacent nodes in $G$ sharing an edge with degrees $d_1$ and $d_2$ respectively. Let $j_1=\lceil \log (d_1+1) \rceil$ and $j_2=\lceil \log (d_2+1) \rceil$, and let $x_1$ and $x_2$ be the integers picked by $p_1$ and $p_2$ respectively using the above algorithm. Then $p_1$ and $p_2$ will never try to host during the same holiday.
\end{lemma}
\begin{proof}

Without loss of generality let $d_1 \leq d_2$. Assume for contradiction that at holiday $t$ both $p_1$ and $p_2$ try to host. This implies that $t=x_1 + a_1 2^{j_1} = x_2 + a_2 2^{j_2}$. But this also means that $x_1 \equiv x_2 \mod 2^{j_1}$, and so when it was the turn of $p_1$ to pick its integer the algorithm could not pick $x_1$ for $p_1$, giving a contradiction.
\end{proof}

\subsection{The Distributed Algorithm}
For the distributed setting we run $\lceil\log (\Delta+1)\rceil$ phases of the BEPS algorithm with the following modification. Starting with $i= \lceil\log (\Delta+1)\rceil$ and going to $i = 0$, during phase $i$ all of the nodes with degree $d$ such that $\lceil \log (d+1) \rceil = i$ will participate in the coloring algorithm for that phase. However, we must restrict the palette of colors for each node $p$ to be comprised only of integers that do not collide (under modulo $2^{i}$) with integers of neighbors of $p$ that already participated in early phases. We emphasize here that the analysis of the BEPS algorithm does not change with this restriction (see~\cite{BEPS12} and~\cite{Johansson99}). The algorithm appears in detail below.

\fbox{
\begin{minipage}{15cm}
{\bf Algorithm -- Distributed degree-bound algorithm}
{\sf
\begin{enumerate}

   \item For $i=\lceil\log (\Delta+1)\rceil$ to $0$
\begin{enumerate}
    \item Let $P_i = \{ p\in P$ such that $\lceil \log(\deg(p)+1) \rceil = i\}$
    \item For each $p \in P_i$ restrict the palette of $p$ to integers that are not equal modulo $2^{i}$ with any integer already picked by a neighbor of $p$.
    \item Run the BEPS algorithm on $P_i$ with the restricted palettes.

\end{enumerate}
\end{enumerate}

{\bf end Algorithm}
}
\end{minipage}
}

Each phase takes $O(\log \Delta + 2^{O(\sqrt{\log \log n})})$ rounds, and we have $O(\log \Delta)$ phases, for a total of $O(\log ^2 \Delta + \log \Delta 2^{O(\sqrt{\log \log n})})$ rounds. The following lemma shows that we will never have conflict.

\begin{lemma}\label{lemma:dist_degree_bound}
Let $p_1$ and $p_2$ be two nodes in $G$ sharing an edge with degrees $d_1$ and $d_2$ respectively. Let $j_1=\lceil \log (d_1+1) \rceil$ and $j_2=\lceil \log (d_2+1) \rceil$, and let $x_1$ and $x_2$ be the integers picked by $p_1$ and $p_2$ respectively using the above algorithm. Then $p_1$ and $p_2$ will never try to host during the same holiday.
\end{lemma}
\begin{proof}
If $j_1 = j_2$ then the BEPS algorithm executed in phase $i$ guarantees that $x_1 \neq x_2$. If $j_1  \neq j_2$ then the proof is the same as that of Lemma~\ref{lemma:sequential_degree_bound}.
\end{proof}

Finally, we have shown the following theorem:

\begin{theorem}
The Distributed degree-bound algorithm guarantees happiness for node
$p\in P$ in every cycle of length bounded above by $2d$, where
$d$ is the degree of $p$.

\end{theorem}

\section{The Dynamic Setting and Open Problems}
\label{s:dynamic}

So far we have assumed that the conflict graph is fixed and does not change throughout the lifetime of the system. However, as we know, relationships are not fixed and new connections may be created or old connections may dissolve. How do our algorithms fare under such conditions? Clearly two nodes that become adjacent and were scheduled to host the same upcoming holiday need to recolor themselves.

In the algorithm of Section~\ref{s:local} all we needed is a valid coloring. So if two nodes $p$ and $q$ that have the same color become connected, one of them, say $p$, needs to find a new color. But this is relatively simple since $p$'s palette should grow by one more color (since $deg(p)$ was increased by $1$). So given the new color, the new periodic schedule for $p$ is derived from the prefix-free encoding of the new color. This means that after at most  $\phi(d) 2^{\log^*d+1}$ rounds after quiescence $p$ will get to host. Note that in this algorithm, if there are $w$ events of adding new neighbors, then the time a node hosts a holiday may be postponed up to $w \cdot  \phi(d) 2^{\log^*d+1}$ rounds. In the event of deletion of edges, presumably there is nothing to be done. However, if this happens too frequently, then the rate of hosting becomes disproportional to the current degree and we will need to recolor the node (again simple even give the smaller palette).

Our Algorithm  of Section~\ref{s:degree_bound} does not fare so well in a dynamic graphs. It is very important in that algorithm to let the higher degree nodes color themselves before the lower degree ones (since the latter's frequency is much higher and they will occupy available slots if colored before the higher degree).

So a main open problem this work presents is whether it is possible to have a degree bound algorithm that works in a dynamic graph. Another issue is whether it is possible to get to the bound of Section~\ref{s:global} of frequency $d+1$ in a periodic or at least succinctly defined manner.

\paragraph{A lower bound conjecture for periodic scheduling.} In Section~\ref{s:degree_bound} we showed an algorithm that achieves a $2d$ upper bound on the period of a vertex with degree $d$. This is in contrast to the non periodic scheduling obtained in Section~\ref{s:global} where we can obtain a $d+1$ guarantee. We conjecture that there is a separation between what is obtainable when one requires periodicity versus the general case, and we leave as an open problem to prove that if one requires a periodic schedule then the best guarantee obtainable is $d+\omega(1)$.

\section*{Acknowledgements}
We thank Amotz Bar-Noy for helpful comments.

\bibliographystyle{amsplain}
\bibliography{paper}

\appendix
\section*{Appendix}

\section{The Complexity of Happiness and Satisfaction} %Apr\`{e}s Moi le D\'{e}luge}
\label{a:deluge}

In this section we consider the problems of maximizing happiness (all children are hosted) or
satisfaction (at least one child is hosted) at a given holiday with no thought of other years.

\subsection{Achieving Maximum Happiness in Hard}\label{ss:max}

Maximizing happiness at a given year means finding the largest set of nodes that can each host all their children.
It is hard to maximize happiness and this follows from the tight relationship with the maximum independent set.

\begin{observation}\label{t:ind}
Maximizing Happiness is $\cal{MAXSNP}$-hard.
\end{observation}

\begin{proof} Consider the conflict graph $G=(P,E)$ as defined in Section~\ref{s:p},
%which we call the {\em
% parent graph}, where $E_p=\{ (p,q) | p,q\in P \ {\rm and} \ (\exists
%c\in C (p,c)\in E \ {\rm and} \ (q,c)\in C\}$.
(where a parent is represented by a node and there is an edge between two nodes if the children of the respective parents are
in a relationship).
It is easy to see that maximizing happiness means finding the maximum
independent set which is $\cal{MAXSNP}$-hard
%MN the point is that it is true for low degree graphs.
even for degree 3 graphs~\cite{bf:95}.
\end{proof}

\subsection{On the Hardness of Being Fair}
\label{sec:hard_fair}
The observation on the hardness of maximizing happiness also implies that any sort of fairness based on maximum happiness will be hard to compute (and hence hard to achieve).
For instance, consider the coalitional game defined by the conflict graph $G=(V,E)$ where for each subset $S \subseteq V$ of nodes the value $v(S)$ is the size of the maximum independent set (MIS) of the subgraph induced by $S$ (which represents the maximum happiness the parents in $S$ can collectively obtain even if all the other nodes give up). Then clearly it is hard to compute $v(S)$. Moreover,  we claim that solution concepts such as the {\em Shapley Value}\footnote{The Shapley Value of a player is based on the expected {\em marginal contribution} of a player to the value of the game when the players are ordered at random; see more details in Osborne and Rubinstein \cite{OsborneR94} and \cite{Naor2005} for its application in the carpool problem.} of this game are hard to compute: take an arbitrary order of the nodes and consider the total (i.e.\ sum of) marginal contributions of the nodes according to this order. It is always equal to the size of the MIS of the graph, since the number of times $MIS(S)$ can grow as $S$ goes from the empty set to the full set is exactly the size of the MIS on the full set of nodes. Therefore any system that approximates the shares according to this definition can also be used to approximate the size of the MIS on the full graph: take the total happiness over a long enough period of time and it should approximate the MIS size. The inapproximability of the MIS problem to a factor of $n^{1-\varepsilon}$ (see~\cite{Hastad99}) implies a large difference between the average rate of hosting for a  random node (which should be close to the fair share) and the size of the MIS divided by $n$.  Thus, we left without a more sophisticated choice than competing with the `chaotic' ``first-come-first-grab" scheme described in the Introduction.

%MN3 reworded
\subsection{Maximum Satisfaction is Linear-Time computable}\label{ss:sat}
We say that parents are satisfied if at least one of their  children
comes home for the holiday.
\begin{definition}
We say that  $p\in P$ is {\em satisfied in a family holiday
 gathering} (orientation) $h$ if $\exists e\in E_p$ such that $e$ is oriented towards $p$ in $h$.
\end{definition}

In contrast to maximizing happiness, the problem of maximizing satisfaction is computationally easy.

\begin{theorem}\label{t:bpm}
Maximum satisfaction can be achieved in linear time.
%MN don't need general maximum matching
% $O(n^{1.5})$.
\end{theorem}

\begin{proof} Maximizing satisfaction means finding the maximum
bipartite matching of the bipartite graph $G'=(P+C, E')$ where the sets of nodes are the parents $P$ and the children $C$ and a parent is connected to its children. This can be done in time
$O(\sqrt{n}|E'|)$ by the Hopcroft-Karp algorithm~\cite{hk:73}. %Since in
%our case every node in $C$ has exactly two edges,
%$|E|=O(n)$. Therefore the maximum satisfaction can be achieved in time
%$O(n^{1.5})$.

%MN
However, a general algorithm for maximum matching in bipartite graphs is an overkill for this problem given that every node in $C$ has exactly two edges: it is possible to find a maximum matching by starting from the parents that have just a single child; each such parent is matched with its child (if two such parents are in-laws, then the satisfied one is decided arbitrarily). Parents who became single child parents since their other children have been matched to their in-laws are treated similarly.
This continues until there are no single child parents. Then all the remaining parents can be satisfied: pick an arbitrary child and match it to a parent. At any point there can be at most one single child parent. Match it to their child.
 \end{proof}

%MN
Note that this solution cannot  be found in a distributed manner quickly, given that maximum matching requires $\Omega(n)$ distance communication in some graphs, such as the cycle of length $n$.

%MN The satisfaction problem is not that interesting
While achieving maximum satisfaction is fast, the solution is not
socially acceptable, since the same parents will be satisfied every
year while others will never be satisfied.
Note that if we want a scheme whereby all parents are guaranteed satisfaction within some cycle of time, then we can guarantee that a parent will not be unsatisfied for more than a year: each child simply alternates and goes one year to its parent and one year to its in-law.

\section{Details for the Elias Omega Code}
\begin{definition}\label{d:elias}
Let $n\in \mathbb{N}$. Denote by $B(n)$ the binary representation of
$n$, i.e.\ $B(n)$ is a string over $\{ 0,1\}$ which is the
representation of $n$ in base $2$ with no leading zeros.
Denote the number of bits in $B(n)$ (the highest power of $2$, $b$
such that $2^b \leq n$) by $|B(n)|$. Let $S$ be a binary
number. Denote by  $LSB(S,k)$ the suffix of $S$ of length $k$, or the
$k$ least significant bits of $S$.

Let $i$ be a positive integer. Denote the empty string by
$\lambda$. Given two strings $u$ and $v$, denote by $u\circ v$ the concatenation of $u$ and $v$. Recursively define a binary string $re(i)$ as follows:
$$re(i)=
\begin{cases}
\lambda & {\rm if\ } $i=1$,
\cr
re(|B(i)|-1)\circ B(i) & {\rm if\ } i>1.
\end{cases}
$$

The {\em Elias omega encoding} of
$i$, is $re(i) \circ 0$ and is denoted by $\omega(i)$.
\end{definition}

{\bf Example:} Consider the Elias omega code of the following numbers:
\begin{enumerate}
\item $i=1$: $re(1)=\lambda$.
{\sl Elias omega code}: $0$.
\item $i=9$:
  $B(9)=1001$. $|B(9)|-1=3$. $B(3)=11$. $|B(3)|-1=1$. Therefore
  $re(9)= re(1) B(3) B(9) = \lambda\ 11\ 1001$\\
{\sl Elias omega code}: $11\ 1001\ 0$.
\item The Elias omega codes of the numbers $1$ to $15$ are:
$0,\ 10\ 0,\ 11\ 0,\ 10\ 100\ 0,\ 10\ 101\ 0,$\\ $10\ 110\ 0,\
10\ 111\ 0,\
  11\ 1000\ 0,\ 11\ 1001\ 0,\ 11\ 1010\ 0,\ 11\ 1011\ 0,\ 11\ 1100\ 0,\
  11\ 1101\ 0$\\ $11\ 1110\ 0,\ 11\ 1111\ 0.$

\end{enumerate}
\begin{properties}\label{p:eomega}
\mbox{}\\
\begin{enumerate} 
\item The Elias omega code is a prefix-free code.
\item The Elias omega code uses $\rho(i) = 1+rb(i)$ bits to represent
  the number $i$, where $rb(i)$ is recursively defined as follows:
$$rb(i)=
\begin{cases}
0 & {\rm if\ } $i=1$,
\cr
\lceil \log(i) \rceil + rb(\lceil \log(i) \rceil -1) & {\rm if\ } i>1.
\end{cases}
.$$
Explicitly, $\rho(i)= 1+ \lceil \log(i) \rceil + \lceil
  \log(\lceil \log(i) \rceil -1) \rceil + \lceil \log(\lceil
  \log(\lceil \log(i) \rceil -1) \rceil - 1) \rceil + \cdots$.
\end{enumerate}
\end{properties}

\end{document}